\documentclass[a4paper,twoside]{amsart}
\usepackage{amsmath}
\usepackage{amsfonts}
\usepackage{amssymb}
\usepackage{amsthm}
\usepackage{newlfont}
\usepackage{graphicx}
\usepackage{amscd}
\usepackage{young}

\theoremstyle{plain}
\newtheorem{Th}{Theorem}[section]
\newtheorem{Cor}[Th]{Corollary}
\newtheorem{Lem}[Th]{Lemma}
\newtheorem{Prop}[Th]{Proposition}
\theoremstyle{definition}

\theoremstyle{remark}
\newtheorem*{Rem}{Remark}
\numberwithin{equation}{section}

\newcommand{\ZZ}{{\mathbb Z}}
\newcommand{\VV}{{\mathbb V}}

\newcommand{\bphi}{\boldsymbol{\phi}}

\newcommand{\bpsi}{\boldsymbol{\psi}}

\begin{document}

\title
{Hermite--Pad\'{e} approximation and integrability}

\author{Adam Doliwa and Artur Siemaszko}

\address{Faculty of Mathematics and Computer Science\\
University of Warmia and Mazury in Olsztyn\\
ul.~S{\l}oneczna~54\\ 10-710~Olsztyn\\ Poland} 
\email{doliwa@matman.uwm.edu.pl, artur@uwm.edu.pl}

%
\date{}
\keywords{Hermite--Pad\'{e} approximation, rational approximation, numerical analysis, discrete integrable systems, the Hirota equation, incidence geometry, multiple orthogonal polynomials, Toda chain equations}
\subjclass[2010]{41A21, 37N30, 37K20, 37K60, 65Q30, 51A20, 42C05}

\begin{abstract}
We show that solution to the Hermite--Pad\'{e} type I approximation problem leads in a natural way to a subclass of solutions of the Hirota (discrete Kadomtsev--Petviashvili) system and of its adjoint linear problem. Our result explains the appearance of various ingredients of the integrable systems theory in application to multiple orthogonal polynomials, numerical algorthms, random matrices, and in other branches of mathematical physics and applied mathematics where the Hermite--Pad\'{e} approximation problem is relevant. We present also the geometric algorithm, based on the notion of Desargues maps, of construction of solutions of the problem in the projective space over the field of rational functions. As a byproduct we obtain the corresponding generalization of the Wynn recurrence.
We isolate the boundary data of the Hirota system which provide solutions to Hermite--Pad\'{e} problem showing that the corresponding reduction lowers dimensionality of the system. In particular, we obtain certain equations which, in addition to the known ones given by Paszkowski, can be considered as direct analogs of the Frobenius identities. We study the place of the reduced system within the integrability theory, which results in finding multidimensional (in the sense of number of variables) extension of the discrete-time Toda chain equations. 

\end{abstract}
\maketitle

\section{Introduction}

\subsection{Numerical algorithms and integrability}
The paper concerns the relationship between numerical algorithms and discrete integrable equations. Our interest will not be however the area of the numerical integration of differential equations but rather the subject related to the \emph{Quotienten--Differenzen} (QD) matrix eigenvalue algorithm, Pad\'{e} approximation or Bose--Chaudhuri--Hocquenghem (BCH)--Goppa decoding algorithm, which can be described~\cite{Moser,GRP,PGR} in terms of the discrete-time Toda chain equations \cite{Toda-67,Symes,Hirota-2dT,PGR-LMP,Hirota-1993}. It is closely related also to various convergence acceleration algorithms whose list together with the correspoding integrable systems can be found in~\cite{CHHL}; see also \cite{NagaiTokihiroSatsuma,BrezinskiHeHuRedivo-ZagliaSun,HeHuSunWeniger}. Due to the structure of determinants involved in the construction, the theory of orthogonal polynomials~\cite{Ismail} is closely related~\cite{Brezinski-PTA-OP} to both the Pad\'{e} approximants and the Toda chain equations (with either continous~\cite{Toda-TL} or discrete time~\cite{Hirota-2dT}), and via this connection the integrable systems appear in the theory of random matrices~\cite{AdlervanMoerbeke,AdlervanMoerbekeVanhaecke}. 

In addition to looking for new examples of such a relationship, one can ask about the general reason explaining its existence.
In the review paper~\cite{Brezinski-CA-XX} one can find the following opinion: \emph{The connection between integrable systems and convergence acceleration algorithms needs to be investigated in more details to fully understand its meaning which is not clear yet}. 
The Toda chain equations can be obtained as a particular integrable reduction of the discrete Kadomtsev--Petviashvili (KP) system by Hirota~\cite{Hirota}, which was in fact published initially under the name of \emph{discrete analogue of a generalized Toda equation}. The Hirota system plays distinguished role in the theory of integrable systems and in their applications~\cite{Miwa,Shiota}. In particular, it is known that the majority of known integrable systems can be obtained as its reductions, appropriate continuous limits~\cite{KNS-rev,Zabrodin} or subsystems \cite{Doliwa-Des-red,doliwakp}. The starting point of the research reported below was raised by the following questions: \emph{What is the role of the Hirota system in the numerical algorithms related to the Pad\'{e} problem? What are their corresponding generalizations, which should contain existing examples as special cases.} 
To support our interest in the problem we mention that in the textbook~\cite{IDS} where the relationship of discrete integrable systems to numerical algorithms based on Pad\'{e} approximants is discussed in detail, it is remarked that extension of the connection of their generalizations \emph{with integrable P$\Delta$E remain largely to be explored}.

\subsection{Description of the main results}
In the present paper we would like to point out the fundamental relation of the Hirota system  to the Hermite--Pad\'{e} approximants, introduced by Hermite~\cite{Hermite,Hermite-P} to prove that Euler's number $e$ is transcendental. Following similar approach Lindemann proved~\cite{Lindemann} analogous result for number $\pi$.
A formal study of Hermite--Pad\'{e} approximation for general systems of functions was initiated by Mahler~\cite{Mahler,Mahler-P}. 

Such generalization of the Pad\'{e} approximants has applications in various areas ranging from the
number theory \cite{Apery}, multiple orthogonal polynomials~\cite{Aptekarev,VanAsche}, random matrices~\cite{Kuijlaars,BleherKuijlaars} to  numerical analysis. 
During last twenty years also the interplay between Hermite--Pad\'{e} approximation theory and integrable systems has been observed in numerous papers~\cite{ BertolaGethmanSzmigielski,ManoTsuda,LopezLagomasinoMedinaPeraltaSzmigielski,AptekarevDerevyaginVanAssche,Filipuk-VanAssche-Zhang,Alvarez-FernandezPrietoManas,AriznabarretaManas}. In particular, in \cite{AptekarevDerevyaginVanAssche}  Hermite--Pad\'{e} approximants for two functions in the framework of discrete integrable systems defined on two-dimensional integer lattice are studied, and the Lax representation for the nearest neighbor recurrence relations is constructed. 

For practitioners of the integrable discrete systems theory it may be therefore interesting to know that within the context of Hermite--Pad\'{e} approximants the Hirota equation and its linear problem  appear in works of Paszkowski~\cite{Paszkowski-H,Paszkowski} and of Della Dora and Di~Crescenzo~\cite{DD-DC-1,DD-DC-2}, but without any mention of their integrability.
It turns out however that only special subclass of solutions of the Hirota system can be obtained this way. The first restriction can be understood by analogy to the fact that only the discrete-time Toda semi-infinite chain equations appear in the Pad\'{e} theory, while in the literature there are known also other solutions like the finite or periodic chain reductions, or infinite chain soliton solutions. 
More importantly, the relevant solutions of the Hirota system are subject to an additional constraint, first observed by Paszkowski~\cite{Paszkowski}. This constraint in the Pad\'{e} approximation case reduces to the discrete-time Toda chain equation itself. In the paper we prove integrability of the constraint in general context thus obtaining a generalization of the discrete-time Toda chain to arbitrary dimension of the lattice. It is important to notice that although the number of discrete variables can be arbitrarily large, the initial boundary value problem is typical for two-dimensional systems, i.e. it is posed on one-dimensional objects. Such "multidimensional" integrable systems are known in the literature, for example the diagonally invariant reduction of the discrete Darboux system~\cite{Doliwa-Santini-sym}. 

In~\cite{AptekarevDerevyaginMikiVanAssche} another equations called discrete multiple Toda lattice have been obtained by generalization to multiple orthogonal polynomials the known relation between the discrete Toda lattice and orthogonal polynomials, where the time variable shows up as a result of an appropriate variation of the measure. The bilinear form of the equations splits also there into the Hirota equations and an additional constraint, almost identical to that introduced by  Paszkowski. As we will demonstrate, both systems of equations are equivalent, see  Section~\ref{sec:integrability} for details.

Another question raised in the paper is the geometric meaning of the Hermite--Pad\'{e} approximation algorithm. The possibility of visualization of arithmetic operations often gives new insight into the problem under consideration; recall, for example, geometric description of the Euclidean algorithm and of the construction of the continued fraction approximation described by Klein~\cite{Klein}. The connection between integrable partial differential equations and the geometry of submanifolds~\cite{Sym,RogersSchief}, deeply rooted in the works of XIXth century geometers~\cite{Bianchi,Darboux-OS,DarbouxIV}, has been transfered to the discrete level~\cite{BP1,BP2,DS-AL,DCN,MQL}, see also \cite{BobSur} for a review. The geometric approach often allows to formulate crucial properties of integrable systems in terms of incidence geometry statements and provides a meaning for various involved calculations. In the paper we give geometric meaning of the construction of solutions to the $m$-dimensional Hermite--Pad\'{e} problem within the context of $(m-1)$-dimensional projective space over the field of rational functions (contained in the field of formal power series). The answer exploits the geometric interpretation of the Hirota discrete KP system given in~\cite{Dol-Des}.

The structure of the paper is as follows. We conclude the Introduction Section by presenting first necessary information about the Hirota system and its integrability, and then we recall basic theory of the Hermite--Pad\'{e} approximants~\cite{DD-DC-1}. In Section~\ref{sec:HP-equations} we study relations between neighbouring elements of the Hermite--Pad\'{e} table, i.e.  the analogs of the classical Frobenius identities. Apart from the known equations we give also some new ones. Then in Section~\ref{sec:geometry-HP} we discuss geometrically the initial boundary-value problem for construction of points of the projective space over the field of rational functions which represent Hermite--Pad\'{e} approximants. Working directly on the level of rational functions we propose also the corresponding counterpart of the Wynn recurrence. Finally, in Section~\ref{sec:integrability} we prove integrability of the Paszkowski reduction of the Hirota equations. 
Our result presents therefore the Hermite--Pad\'{e} approximants within the more general context of the integrable systems theory thus explaining appearance of various ingredients of the theory in application to multiple orthogonal polynomials, numerical algorithms, random matrices, and in all other branches of mathematical physics and applied mathematics where the Hermite--Pad\'{e} approximation problem is relevant.

\subsection{The Hirota equation} \label{sec:Hirota}

The standard form of Hirota's discrete Kadomtsev--Petviashvili (KP) system of equations \cite{Hirota,Miwa} reads as follows
\begin{equation} \label{eq:H-M}
\tau_{(i)}\tau_{(jk)} - \tau_{(j)}\tau_{(ik)} + \tau_{(k)}\tau_{(ij)} = 0,
\qquad 1\leq i< j <k ;
\end{equation}
here $\tau$ is unknown function of discrete variables $(n_1, n_2, n_3, \dots  )$, $n_k\in\ZZ$, with values in a field; usually the real or complex numbers, but other possibilities like the finite fields were also studied in the literature~\cite{BialeckiDoliwa}. In the above and in all the paper we use the short-hand notation with indices in brackets meaning shifts in discrete variables, for example 
\begin{equation} \label{eq:shifts}
\tau_{(\pm i)}(n_1, \dots , n_i, \dots  ) = \tau(n_1, \dots , n_i \pm 1, \dots );
\end{equation}
accordingly, $\tau_{(ij)}$ means increasing by $1$ both arguments $n_i$ and $n_j$.
\begin{Rem}
Initially the Hirota system was considered for $m=3$ discrete variables only, which results in a single equation. An important feature of the equations is that the number of discrete variables can be grater then three and such augumentation does not lead to any restriction on the solution space, i.e. the additional variables can be considered as parameters of commuting symmetries of the original equation. Such a \emph{multidimensional consistency}, which was observed also for other systems like for example the discrete Darboux equations~\cite{MQL,TQL}, has been placed as the central concept of the modern theory of discrete integrable systems in~\cite{FWN-lB,ABS}.
\end{Rem}
\begin{Rem}
The Hirota system has actually more involved (affine Weyl group of type $A$) symmetry structure~\cite{Dol-AN} than that described above. 
\end{Rem}

To apply techniques of the integrable systems theory to construct solutions of the Hirota system it is important to represent it as compatibility condition of a linear system, whose standard form~\cite{DJM-II} reads as follows
\begin{equation} \label{eq:lin-dKP}
\bpsi_{(i)} - \bpsi_{(j)} = \frac{\tau \tau_{(ij)}}{\tau_{(i)} \tau_{(j)}} \bpsi ,  \qquad 1\leq i < j ,
\end{equation}
where the so called wave function $\bpsi$ takes values in a vector space $\VV$.
Equivalently one can consider the adjoint linear problem
\begin{equation} \label{eq:lin-dKP*}
\bpsi^*_{(j)} - \bpsi^*_{(i)} = \frac{\tau \tau_{(ij)}}{\tau_{(i)} \tau_{(j)}}  \bpsi^*_{(ij)} , 
\qquad1\leq i < j ,
\end{equation}
where the adjoint wave function $\bpsi^*$ takes values in the dual space $\VV^*$; this point of view is important in construction of the Darboux transformations~\cite{Nimmo-KP} of the system.
\begin{Rem}
Notice that the Hirota system is invariant with respect to reversal of directions of all discrete variables, and upon the direction reversal the linear and adjoint problems are exchanged.
\end{Rem}

There exists~\cite{Saito-Saitoh} a duality between the linear problems and the Hirota equation itself. An easy calculation shows that the vector-valued functions 
	\begin{equation} \label{eq:psi-phi}
	\bphi = \tau \bpsi, \qquad \bphi^* = \tau \bpsi^* ,
	\end{equation}
	satisfy the following (so called bilinear) form of the linear problem and of its adjoint, where we still keep the assumption that $i<j $
	\begin{align} \label{eq:lin-bil}
	\tau_{(j)} \bphi_{(i)} - \tau_{(i)} \bphi_{(j)} & =  \tau_{(ij)} \bphi,\\
	\label{eq:ad-lin-bil}
	\bphi^*_{(j)}\tau_{(i)}  -  \bphi^*_{(i)}\tau_{(j)} & =  \bphi^*_{(ij)} \tau.
	\end{align}
\begin{Cor} \label{cor:quad-phi}
The above duality implies that scalar components $\phi_\ell$ (or $\phi^{*}_{\ell}$) of the new wave function (or of its adjoint) satisfy the Hirota system \eqref{eq:H-M} as well.
\end{Cor}
\begin{proof}
To show directly that $\phi^*_\ell$ satisfies equation \eqref{eq:H-M} 
multiply the $\ell$th component of \eqref{eq:ad-lin-bil} by $\phi^*_{\ell(k)}$, where we assume that $i<j<k$. Then subtract from it the $ik$ version of \eqref{eq:ad-lin-bil} multiplied by $\phi^*_{\ell(j)}$, and add the $jk$ version multiplied by $\phi^*_{\ell(i)}$. The terms on the left side of equality cancel out, and the right side divided by $\tau$ results in
\begin{equation*}
\phi^*_{\ell(i)}\phi^*_{\ell(jk)} - \phi^*_{\ell(j)}\phi^*_{\ell(ik)} + \phi^*_{\ell(k)}\phi^*_{\ell(ij)} = 0,
\qquad 1\leq i< j <k .
\end{equation*}	
\end{proof}

\begin{Rem}
	As it was shown in \cite{Dol-Des} the ratio of two scalar solutions of the linear system \eqref{eq:lin-dKP} (thus also of \eqref{eq:lin-bil}) satisfies the so called Schwarzian form of the discrete Kadomtsev--Petviashvili (dSKP) system \cite{FWN-Capel,BoKo-N-KP}, which will be relevant in Section~\ref{sec:geometry-HP}. By the reversal of the directions symmetry, discussed above, analogous result is true also for the adjoint linear problems.
\end{Rem}

The initial boundary value data for the Hirota system were given on the geometric level in \cite{Dol-Des} on the basis of the relation of Desargues maps to multidimensional quadrilateral lattices~\cite{DCN,MQL}. The data consist of a system of two dimensional lattices of planar quadrilaterals (one such a lattice in the basic case of three discrete variables). See also \cite{Doliwa-Des-red} for explanation of the relation on the level of the root lattices and corresponding Weyl groups.

\subsection{Hermite--Pad\'{e} approximants} \label{sec:H-P}
The classical Hermite--Pad\'{e} problem and its solution can be stated as follows \cite{DD-DC-1}. Given $m$ elements $(f_1,\dots , f_m)$ of the algebra $\Bbbk[[x]]$ of formal series in variable $x$ with coefficients in the field $\Bbbk$, 
\begin{equation}
	f_i(x) = \sum_{j=0}^\infty f^i_j x^j, 
\end{equation}
and given an element $n=(n_1,\dots , n_m)$ of $\ZZ_{\geq -1}^m$, where we also write $|n|= n_1 + \dots + n_m$, a \emph{Hermite--Pad\'{e} form of degree $n$} is every system of polynomials $(Y_1, \dots , Y_m)$ in $\Bbbk[x]$, not all equal to zero, with corresponding degrees $\deg Y_i \leq n_i$, $i=1,\dots , m$ (degree of the zero polynomial by definition equals $-1$), and such that 
\begin{equation}
Y_1(x)f_1(x) + \dots + Y_m(x) f_m(x) = x^{|n|+m-1}\Gamma(x)
\end{equation}
for a series $\Gamma \in \Bbbk [[x]]$ with non-negative powers of $x$.

Define the matrix $\mathcal{M}(n)$ of $(|n|+m-1)$ rows and $(|n|+m)$ columns 
\begin{equation}
\mathcal{M}(n) = \left( \begin{array}{cccccccc}
f^1_0 &     & 0 & & & f^m_0 &   & 0 \\
f^1_1 & \ddots    &   & \cdots & \cdots & f^m_1 & \ddots  &   \\
\vdots &   &  f^1_0 & & & \vdots &    &  f^m_0 \\
\vdots &   &  \vdots & \cdots &\cdots & \vdots &  & \vdots\\
f^1_{|n|+m-2} & \cdots & f^1_{|n|+m - n_1 -2} & & & f^m_{|n|+m-2} & \cdots & f^m_{|n|+m - n_m -2}
\end{array} \right) ;
\end{equation}
its columns are divided into $m$ (possibly empty) groups, the $i$th group is composed out of $n_i +1$ columns depending on the coefficients of $f_i(x)$ only.
By supplementing $\mathcal{M}(n)$ at the bottom by the line
\begin{equation*}
\left( f_1, xf_1, \dots , x^{n_1} f_1, \cdots , \cdots , f_m, x f_m, \dots , x^{n_m} f_m \right)
\end{equation*}
and calculating the determinant of the resulting square matrix in two ways we obtain the following identity
\begin{equation} \label{eq:Z-f-Delta}
Z_1(x) f_1(x) + \dots + Z_m(x) f_m(x) = x^{|n|+m-1} \sum_{j=0}^\infty \Delta^{(j-1)} x^j, 
\end{equation}
where each $Z_k\in \Bbbk[x]$ is a polynomial of degree not exceeding $n_k$, $k=1,\dots m $, given explicitly as the determinant
\begin{equation} \label{eq:Z-det}
Z_k(x) = \left|
 \begin{smallmatrix}
f^1_0 &     & 0 & & & f^k_0 & & 0 && & f^m_0 &   & 0 \\
f^1_1 & \ddots    &   & \cdots & \cdots & f^k_1  &\ddots &  &\cdots &\cdots & f^m_1 & \ddots  &   \\
\vdots &   &  f^1_0 & & & \vdots && f^k_0 & && \vdots &     &  f^m_0 \\
\vdots &   &  \vdots & \cdots &\cdots &&&&\cdots &\cdots & \vdots &  & \vdots\\
f^1_{|n|+m-2} & \cdots & f^1_{|n|+m - n_1 -2} & & &f^k_{|n|+m-2}& \cdots & f^k_{|n|+m-n_k - 2}& & &  f^m_{|n|+m-2} & \cdots & f^m_{|n|+m - n_m -2} \\
0 & \cdots & 0 & \cdots & \cdots &1 & \cdots & x^{n_k} & \cdots &\cdots & 0 & \cdots & 0
\end{smallmatrix} \right|
\end{equation}
of the matrix $\mathcal{M}(n)$ supplemented at the bottom by the line 
\begin{equation}
\label{eq:Xk}
X_k = (0,\dots, 0, \dots \; \dots , 1, x , \dots , x^{n_k}, \dots \; \dots , 0, \dots ,0),
\end{equation}
consisting of zeros except for the $k$th block of the form $1, x, \dots , x^{n_k}$. Moreover
$\Delta^{(j)}(n)$ is the determinant of the matrix $\mathcal{M}(n)$ supplemented by the line
\begin{equation} \label{eq:bN}
b_{|n|+m+j} = \left( f^1_{|n|+m+j}, f^1_{|n|+m+j-1}, \dots , f^1_{|n|+m+j-n_1}, \cdots , \cdots , f^m_{|n|+m+j},  \dots , f^m_{|n|+m+j-n_m} \right),
\end{equation} 
as the last row;
in particular, the coefficient at $x^{|n|+m-1}$ of $Z_1(x) f_1(x) + \dots + Z_m(x) f_m(x) $ equals 
\begin{equation}
\Delta(n) = \Delta^{(-1)}(n) = \left| \begin{array}{cccccccc}
f^1_0 &     & 0 & & & f^m_0 &   & 0 \\
f^1_1 & \ddots    &   & \cdots & \cdots & f^m_1 & \ddots  &   \\
\vdots &   &  f^1_0 & & & \vdots &    &  f^m_0 \\
\vdots &   &  \vdots & \cdots &\cdots & \vdots &  & \vdots\\
f^1_{|n|+m-2} & \cdots &   & & &  & \cdots & f^m_{|n|+m - n_m -2} \\
f^1_{|n|+m-1} & \cdots & f^1_{|n|+m - n_1 -1} & & & f^m_{|n|+m-1} & \cdots & f^m_{|n|+m - n_m -1}
\end{array} \right| .
\end{equation}
The above system of polynomials $ (Z_1, \dots , Z_m)$ is called  \emph{the canonical Hermite--Pad\'{e} form of degree~$n$}. 

The polynomials $Z_k$, $k=1,\dots, m$, depend not only of the variable $x$, but also on the full set of discrete variables (like the determinants $\Delta(n)$) so we have $Z_k(n,x)=Z_k(n_1,\dots , n_m,x)$. We will however often skip the variables if it does not lead to ambiguity. We denote also shifts in the discrete variables as in equation \eqref{eq:shifts}.

Because the leading term of the polynomial $Z_k$ reads
\begin{equation} \label{eq:leading-term-Z}
Z_k(n,x) = (-1)^{(n_{k+1} + \dots + n_m) - m+k}\Delta_{(-k)}(n) x^{n_k} + \dots,
\end{equation}
therefore if for all $n$ the determinant $\Delta(n)$ does not vanish then the polynomials $Z_k$ are of the maximal order. Such a system of series $(f_1, \dots , f_m)$ is called \emph{perfect} \cite{Mahler-P}, which we assume in the sequel.

\begin{Rem}
	In our paper we consider only the so called type I approximants of the system $(f_1,\dots ,f_m)$. For closely reated type II and mixed type approximants see, for example \cite{Mahler-P}. Similar terminology applies to the corresponding types of multiple orthogonal polynomials, for details see \cite{Ismail,VanAsche}. 
\end{Rem}

\section{Linear and nonlinear equations behind the Hermite--Pad\'{e} approximation problem}
\label{sec:HP-equations}
In this Section we study integrable equations satisfied by the polynomials $Z_\ell$ and the function $\Delta$. After discussing first the generic case we consider then equations in two discrete variables only, which are well known from the theory of Pad\'{e} approximants. They will be needed to formulate initial-boundary value problem for the full set equations in arbitraty number of discrete variables. They will also motivate our interest in providing their counterparts in the Hermite--Pad\'{e} case.
\begin{Rem}
We will need the following determinantal identity, which within the integrable systems community known as the Jacobi identity~\cite{Hirota-book} or Dodgson condensation rule, but within the Pad\'{e} approximation community known as the Sylvester identity~\cite{BakerGraves-Morris}, who generalized works of Jacobi. 

\noindent \emph{Let $A$ be a square matrix, chose two rows with indices $i_1<i_2$ and two columns with indices $j_1<j_2$. Let $A^{i,j}$ denote the matrix with row $i$ and column $j$ deleted. Also let $A^{i_1 i_2, j_1 j_2}$ denote the matrix $A$ with rows $i_1, i_2$ and columns $j_1, j_2$ deleted, then \begin{equation*}
\det A  \; \det A^{i_1 i_2, j_1 j_2} = \det A^{i_1 , j_1 } \; \det A^{i_2, j_2} - \det A^{i_2 , j_1 } \; \det A^{i_1, j_2} .
\end{equation*}}
\end{Rem}

\subsection{Equations in all discrete variables}
\label{sec:mult-eq-HP}
Application of the Sylvester identity to the determinant $Z_{\ell(ij)}$, $\ell = 1,\dots ,m$,  with two bottom rows and the last columns of the $i$th and $j$th blocks gives \cite{DD-DC-2}
\begin{equation} \label{eq:Z-Delta}
Z_{\ell(ij)} \Delta = Z_{\ell(j)} \Delta_{(i)} - 
Z_{\ell(i)} \Delta_{(j)}, \qquad i<j , \qquad \ell = 1, \dots , m.
\end{equation}
By comparing the above system with the symmetric form of the adjoint linear problem \eqref{eq:ad-lin-bil} we identify $\Delta$ with the $\tau$-function and the vector $(Z_1,\dots ,Z_m)$ as the adjoint wave function $\bphi^*$.
By eliminating either $Z_\ell$ or $\Delta$ from equations \eqref{eq:Z-Delta} for three pairs of indices $i<j$, $i<k$ and $j<k$, exactly like in the proof of Corollary~\ref{cor:quad-phi}, we obtain the well known equations.
\begin{Prop}
	The basic determinant $\Delta$ and the canonical Hermite--Pad\'{e} forms $Z_\ell$ satisfy the standard bilinear Hirota equations
	\begin{align}
	\label{eq:H-Delta}
	\Delta_{(ij)} \Delta_{(k)} - \Delta_{(j)} \Delta_{(ik)} +
	\Delta_{(i)} \Delta_{(jk)} = & 0, \qquad 1\leq i < j < k \leq m,\\
	\label{eq:H-Z}
	Z_{\ell(ij)}Z_{\ell(k)} - Z_{\ell(ik)}Z_{\ell(j)} + Z_{\ell(jk)}Z_{\ell(i)} = & 0,  \qquad \quad \ell = 1,\dots , m.
	\end{align} 
\end{Prop}
\begin{Rem}
To make the paper self-contained we sketch yet another proof of the Proposition using determinantal identities.
Let us consider the determinant of the matrix $M_{(ijk)}(n)$ supplemented
at the bottom by the row with all zeros except $1$ at the last column of the $k$th group. By applying the
Sylvester identity with respect to last two rows and the last columns of the blocks $i$ and $j$ we obtain equation~\eqref{eq:H-Delta}. 
In order to prove this way equation \eqref{eq:H-Z} we discuss only the case $\ell\neq  i, j, k$, leaving other instances to the Reader. Consider determinant of
the matrix $M_{(ijk)}(n)$ but with the last row replaced by $X_{\ell(i,j,k)}(n,x)$, and supplemented at the bottom by
the row with all zeros except $1$ at the last column of the $k$th group. Then apply the Sylvester identity
with respect to the third row from the bottom and the last row, and the last columns of the blocks $i$ and $j$. 
\end{Rem}
\begin{Rem}
In application to Hermite--Pad\'{e} approximation problem equation \eqref{eq:H-Delta} can be found in \cite{Paszkowski-H} without any relation to integrability, but already for an arbitrary number of variables (actually, we refer to that work following \cite{BakerGraves-Morris}, where equations \eqref{eq:H-Delta} were given as equations (5.20) for $j=0$). It seems that in the context of Hermite--Pad\'{e} approximants equation \eqref{eq:H-Z} has not been considered before.
\end{Rem}

The idea of using the fundamental equation~\eqref{eq:Z-Delta} to fill the table of canonical Hermite--Pad\'{e} forms was used in~\cite{DD-DC-2} and can be implemented as follows. The set of canonical forms of degree $n$ with $|n|+m-1=K$ will be called $K$th level of the Hermite--Pad\'{e} table. The corresponding functions $\Delta$ of this level can be found as coefficients at $x^{K}$ of $Z_1 f_1 + \dots + Z_m f_m$ (lower order terms vanish, see equation \eqref{eq:Z-f-Delta}). Notice that equation~\eqref{eq:Z-Delta} implies that the forms of $K$th level can be found (up to boundary elements) using forms of level $(K-1)$ and functions $\Delta$ of level $(K-1)$ and $(K-2)$. The algorithm to fill out the Hermite--Pad\'{e} table up to level $N$ requires truncated series $f_k(x)$, and reads then as follows	
	\begin{itemize}
		\item[1.] Start with $\Delta(-1,\dots ,-1) = 1$, $Z_j(-1,\dots , -1,x) = 0$;
		\item[2.] For $K=0,1,2,\dots, N$:
		\begin{itemize}
			\item[a)] From equation \eqref{eq:Z-det} read $Z_j(-1, \dots , n_i=K, \dots , -1,x)= \delta_{ij}(f^i_0 x)^{K}$;
			\item[b)] Find the remaining polynomials $Z_\ell$ of the level $K$ by using equation \eqref{eq:Z-Delta};
			\item[c)] Find all $\Delta$ of this level using equation \eqref{eq:Z-f-Delta}. 
		\end{itemize}
	\end{itemize}

\begin{Rem}
Notice that to find $Z_\ell$ at the step 2b) any possible choice of indices $i<j$ such that $n_i, n_j\geq 0$ gives the correct result, which is due to the compatibility of equations \eqref{eq:Z-Delta}, i.e. the integrability of the Hirota system.
\end{Rem}
\begin{Rem}
The above algorithm is not of the form of the difference equation because at the step 2c) the higher order coefficients of the series are drown into the solution. At this step (i.e. level $K$) one needs polynomial approximations of the series up to terms of the order $K$.	
\end{Rem}
\begin{Rem}
Equation \eqref{eq:Z-Delta} can be also obtained by using the uniqueness (up to a scalar factor) of the solution of the Hermite--Pad\'{e} approximation problem~\cite{Paszkowski}. Such a reasoning is typical for analytic techniques of the soliton theory such as the finite-gap integration technique~\cite{Krichever-1} or the non-local $\bar{\partial}$ dressing method~\cite{AblowitzBarYaacovFokas,BogdanovManakov} and its descendants (like the Riemann--Hilbert or the inverse spectral transform techniques~\cite{NMPZ,AblowitzSegur}). 
	A part of the technique is derivation of solutions of the nonlinear equations by taking coefficients of expansion of solutions of the corresponding linear problems. For example, in equation \eqref{eq:Z-Delta} let us take $\ell = k$ different form $i,j$ and use the leading term coefficient formula \eqref{eq:leading-term-Z} to get \eqref{eq:H-Delta}.
\end{Rem}

\subsection{Equations on the initial planes} \label{sec:initial-Wynn}
To formulate the initial boundary value problem we fix two indices $i<j$, and in this Section we consider the special case when $n_k = -1$ for all indices $k\neq i,j$. For such initial planes the matrix $\mathcal{M}(-1,\dots , n_i, \dots , n_j, \dots , -1)$ consists of two blocks only --- those indexed by $i$ and $j$.  Correspondingly, by equation \eqref{eq:Z-det} we have $Z_\ell(-1,\dots,n_i,\dots,n_j,\dots,-1,x)\equiv 0$ for $\ell\neq i,j$. Without loss of generality one could fix $i=1$ and $j=2$, therefore we are left  with the Pad\'{e} problem for the quotient of two power series in the formalism of bigradients~\cite{Gragg}.

Application of the Sylvester identity for the two bottom rows and the first column of the block $i$ and the last column of the block $j$ of the determinants $Z_{i(ij)}$ and $Z_{j(ij)}$ (as other polynomials vanish there) gives on the initial planes
\begin{equation} \label{eq:F-2}
Z_{\ell(ij)}\Delta_{(-j)} = x Z_{\ell}\Delta_{(i)} - Z_{\ell(i)} \Delta, \qquad \ell=1,\dots ,m.
\end{equation}
Doing the same but for last column of the block $i$ and the first column of the block $j$ we obtain 
\begin{equation} \label{eq:F-3}
Z_{\ell(ij)}\Delta_{(-i)} = - x Z_{\ell}\Delta_{(j)} + Z_{\ell(j)} \Delta, \qquad \ell=1,\dots ,m.
\end{equation}

Equations \eqref{eq:F-2}, \eqref{eq:F-3} and  \eqref{eq:Z-Delta} are part of the well known Frobenius identities for the Pad\'{e} problem. The remaining ones, for example
\begin{align} \label{eq:F-ZZ}
Z_\ell^2 = & Z_{\ell(i)} Z_{\ell(-i)} +  Z_{\ell(j)} Z_{\ell(-j)} ,\\
\label{eq:F-XDZ}
xZ_\ell \Delta = & Z_{\ell(i)} \Delta_{(-i)} +  Z_{\ell(j)} \Delta_{(-j)}, \\ \label{eq:F-DD}
\Delta^2 = & \Delta_{(i)} \Delta_{(-i)} +  \Delta_{(j)} \Delta_{(-j)},
\end{align}
can be obtained from them. 
\begin{Rem}
Equation \eqref{eq:F-DD} within the theory of integrable systems was derived in \cite{Hirota-2dT} as the discrete-time form of the Toda chain equations~\cite{Toda-TL}. The Pad\'{e} problem provides only special solutions of the equation. Such solutions are supported in the first quadrant of the two dimensional integer lattice. In theory of integrable systems also other types of solutions are known, for example the finte chain, periodic or quasi-periodic solutions, the multisoliton solutions.
\end{Rem}
The above formulae imply also the ``missing identity of Frobenius" by Wynn~\cite{Wynn} for the ratio $R(x)=Z_i(x)/ Z_j(x) \in \Bbbk(x)$ in the field of rational functions over $\Bbbk$. We will discuss the identity in Section~\ref{sec:gen-Wynn} together with its geometric interpretation~\cite{Doliwa-Siemaszko-W}.

\subsection{The Paszkowski constraint} \label{sec:Paszkowski}
In \cite{Paszkowski} it was shown by Paszkowski, using analyticity and uniqueness properties of the solution of the Hermite--Pad\'{e} problem, that the canonical Hermite--Pad\'{e} form satisfies the following equation
\begin{equation}
\label{eq:Paszkowski-Frobenius}
x {Z_\ell} \Delta = Z_{\ell(1)} \Delta_{(-1)} + \dots + Z_{\ell(m)} \Delta_{(-m)}, \qquad \ell = 1, \dots , m, 
\end{equation}
the corresponding generalization of the Frobenius identity~\eqref{eq:F-XDZ}. In Section~\ref{sec:integrability} we will consider the admissibility of the above Paszkowski condition in the wider context of integrable reductions of the Hirota system and its linear problem.
\begin{Rem}
	Below we demonstrate validity of the equation by using determinantal identities.
\end{Rem}
\begin{Rem}
The above identity (and its extensions \cite{Paszkowski}) can be used to obtain a particular element of the Hermite--Pad\'{e} table. Instead of filling the table level by level one can reach the element by moving along a specific path, which speeds up the calculation.
\end{Rem}

From the leading term of the main approximation condition \eqref{eq:Z-f-Delta} in equation~\eqref{eq:Paszkowski-Frobenius} one can find~\cite{Paszkowski-H}, the corresponding analog of the Frobenius identity \eqref{eq:F-DD}
\begin{equation} \label{eq:Paszkowski-D}
\Delta^2 = \Delta_{(1)} \Delta_{(-1)} + \dots + \Delta_{(m)} \Delta_{(-m)}, 
\end{equation}
derived later in \cite{BakerGraves-Morris} by using determinantal identities. The $\Delta \leftrightarrow Z_\ell$ symmetry of the linear problems \eqref{eq:Z-Delta} and \eqref{eq:Paszkowski-Frobenius} suggests that an equation analogous to \eqref{eq:Paszkowski-D} should  be satisfied also by $Z_\ell$.  
\begin{Prop} \label{prop:ZZ}
	The components $Z_\ell$ of the canonical Hermite--Pad\'{e} form satisfy the quadratic equation
\begin{equation}
\label{eq:Paszkowski-Z}
{Z_\ell}^2 = Z_{\ell(1)} Z_{\ell(-1)} + \dots + Z_{\ell(m)} Z_{\ell(-m)}, \qquad \ell = 1, \dots , m, 
\end{equation}	
which is the analog of the Frobenius identity \eqref{eq:F-ZZ}.
\end{Prop}
\begin{proof}
We follow the idea of~\cite{BakerGraves-Morris} applied there in derivation of equation \eqref{eq:Paszkowski-D}. We will use notation introduced in Section~\ref{sec:H-P}, where
\begin{equation}
Z_\ell = \left| \begin{matrix} \mathcal{M} \\X_\ell \end{matrix} \right| , \qquad \Delta = \left| \begin{matrix} \mathcal{M} \\ b_N \end{matrix} \right|, \qquad  N=|n| + m -1. 
\end{equation}
As a warming-up let us first derive in this way equation \eqref{eq:Paszkowski-Frobenius}, as promised. Consider the determinant 
\begin{equation*}
\mathcal{D}_\ell = \left| \begin{array}{c|c} \mathcal{M} & 0_{N+1} \\
b_{N} & \mathcal{M} \\ \hline
X_\ell & x\, X_\ell \\
0^N_{N+1} & \mathcal{M}
\end{array} \right| = (-1)^N  x\Delta Z_{\ell} ,
\end{equation*}
where $0_{N+1}$ is a row vector with $N+1$ zeros and  $0^N_{N+1}$ is $N\times (N+1)$ array of zeros. The equality follows by first using elementary row operations to convert the upper right block to a block of zeros, then moving $(N+2)$-th row to the bottom, and finally by applying the generalized Laplace expansion with respect to the first $(N+1)$ columns. 

By elementary column operations we can create zeros in all of the top $(N+2)$ entries in every one of the colums in the right half of $\mathcal{D}_\ell$ except for the last ones in each of $m$ blocks. In the generalized Laplace expansion with respect to first $N+2$ rows the non-vanishing terms arise from determinants formed by the rows and the first $N+1$ columns supplemented by the last ones in each of the blocks. After transpositions, which produce the correct sign, such a determinant formed from $k$-th block gives the polynomial $Z_{\ell(k)}$, while the corresponding multiplier being the determinant of order $N$ gives $\Delta_{(-k)}$, what results in
\begin{equation*}
\mathcal{D}_\ell = (-1)^N  \left( Z_{\ell(1)} \Delta_{(-1)} + \dots +  Z_{\ell(m)} \Delta_{(-m)} \right).
\end{equation*}

Let us proceede to the proof of the main identity \eqref{eq:Paszkowski-Z}. Denote by $\mathcal{M}^\prime$ the matrix $\mathcal{M}$ with the last row $b_{N-1}$ removed, and consider the following determinant
\begin{equation}
(-1)^{N-1}  Z_{\ell}^2= \left| \begin{array}{c|c} 
\mathcal{M}& \begin{matrix} 0_{N+1} \\ \mathcal{M}^\prime  \end{matrix}\\
X_\ell & x\, X_\ell \\ \hline
b_{N} & b_{N-1} \\ 
0^N_{N+1} & \begin{matrix} \mathcal{M}^\prime \\X_\ell \end{matrix}
\end{array} \right| = (-1)^{N-1}  \left(  Z_{\ell(1)} Z_{\ell(-1)} + \dots + Z_{\ell(m)} Z_{\ell(-m)} \right).
\end{equation}
The left equality follows again from converting the upper right block to a block of zeros, but then moving $(N+2)$-th row to the penultimate one. To get the right equality we apply the column operations as before, but first we exchange the rows $(N+1)$ and $(N+2)$.
\end{proof}
\begin{Rem}
To complete presentation of the analogs of the basic Frobenius identities we add to the list equations
	\begin{equation} \label{eq:Frobenius-HP-Z}
	xZ_\ell \Delta_{(i)} = Z_{\ell(i)} \Delta + \sum_{j=1}^m \mathrm{sgn}(j-i) Z_{\ell(ij)} \Delta_{(-j)}, \qquad i,\ell = 1,\dots , m, 
	\end{equation}
\begin{equation*}
\mathrm{sgn}(k) = \begin{cases}
1 & k>0, \\
0 \qquad  \text{for} & k=0,\\
-1&  k< 0,
\end{cases}
\end{equation*}
	which generalize the Frobenius identities \eqref{eq:F-2}-\eqref{eq:F-3}. The Reader can check that the additional linear problems \eqref{eq:Frobenius-HP-Z} can be also obtained using the analytic properties of the polynomials $Z_\ell$ and assumption about the uniqueness of the solution. We will also derive the equations in  more general context in Section~\ref{sec:integrability}.	
\end{Rem}

\section{Geometry of the Hermite--Pad\'{e} approximation}
\label{sec:geometry-HP}
\subsection{The collinearity condition and the local geometric construction}
According to the geometric meaning~\cite{Dol-Des} of the Hirota system let us interpret the canonical Hermite--Pad\'{e} form of degree $n$ as \emph{homogeneous coordinates} 
\begin{equation}
P(n,x) = [Z_1(n,x) : \ldots : Z_m(n,x)]
\end{equation} of the point $P(n,x)$ in $(m-1)$ dimensional projective space over the field of rational functions $\Bbbk(x)$. 
\begin{Rem}
Actually, in the context of Hermite--Pad\'{e} approximation it also makes sense to consider the field of Laurent series $\sum_{i\geq i_0}^\infty a_i x^i$ with indices of the coefficients bounded from below, which contains $\Bbbk(x)$ as a proper subfield.	
\end{Rem} 

Equations~\eqref{eq:Z-Delta} of the linear system mean that the point $P(n,x)$ and its $m$ backward neighbours $P_{(-i)}(n,x) = P(n_1, \dots , n_i -1, \dots , n_m,x)$, $i=1,\dots , m$, are collinear; let us denote the line by $L(n,x)$. The local construction step is based on the observation that the collinearity constraints imply that any five out of the six points $P_{(i)}$, $P_{(j)}$, $P_{(k)}$, $P_{(ij)}$, $P_{(ik)}$, $P_{(jk)}$ determine the sixth one. The points and four corresponding lines $L_{(ij)}$, $L_{(ik)}$, $L_{(jk)}$, $L_{(ijk)}$ form the so called Veblen configuration (here as usual all the three indices $i,j,k$ are distinct), see Figure~\ref{fig:R-Veblen-PL}. 
\begin{figure}[h!]
	\begin{center}
		\includegraphics[width=6cm]{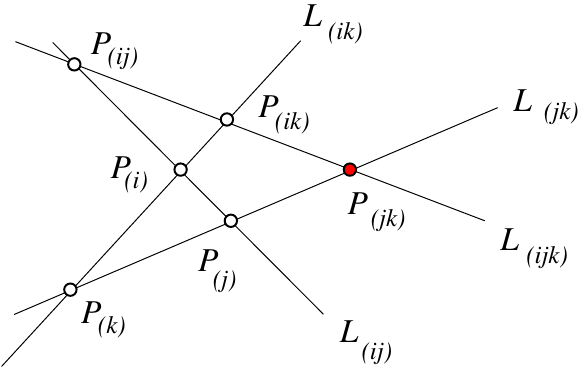}
	\end{center}
	\caption{Local geometric basic construction step of the projective solution to the Hermite--Pad\'{e} problem}
	\label{fig:R-Veblen-PL}
\end{figure}
Notice that the construction involves three dimensions of the discrete parameters. For $m>3$ there is not unique way to contruct the points but due to the multidimensional consistency of the Hirota system all the ways give the same result~\cite{Dol-Des}. Geometrically, the consistency is equivalent to the Desargues theorem of projective geometry and its multidimensional generalizations~\cite{Dol-AN}.

\subsection{The generalized Wynn recurrence} \label{sec:gen-Wynn}
The homogeneous coordinates of a point in projective space are not given uniquely. 
 We can uniquely represent the points $P(n)$ of the projective $(m-1)$ dimensional space in terms of \emph{non-homogeneous coordinates} $(R_1(n,x),\dots ,R_{m-1}(n,x))$, where $R_i(n,x) = Z_{i+1}(n,x)/Z_1(n,x) \in \Bbbk(x)$, $i=1,\dots ,m-1$, are rational functions. Recall that a fixed system of non-homogeneous coordinates in projective space does not cover the whole space. The corresponding hyperplane at infinity, where the above non-homogeneous coordinate system cannot be used, is given by $Z_1 = 0$. In such a case one can change the system to another one by using other non-zero polynomial $Z_j$, $j\neq 1$, as the denominator. 

In \cite{Dol-Des} it was shown (for arbitrary field, including the non-commutative case as well) that the resulting non-homogeneous coordinates satisfy the discrete Schwarzian Kadomtsev--Petviashvili (dSKP) equation; see also \cite{FWN-Capel,BoKo-N-KP}. To make the paper self-contained let us recall the result and its proof (given here in turn for the adjoint linear problem). We will give also another geometric interpretation of the equation in terms of the so called quadrangular set of points \cite{DoliwaKosiorek} which is valid for arbitrary field; see also \cite{KoSchief-Men} for the complex field interpretation in the context of the Menelaus theorem and of the Clifford configuration of circles in inversive geometry. 
\begin{Prop}
Given two non-trivial scalar solutions $Z$ and $\tilde{Z}$ of the linear system \eqref{eq:Z-Delta} then for any triple $i<j<k$ of indices the rational function $R =Z / \tilde{Z} \in \Bbbk(x)$ satisfies the dSKP equation
	\begin{equation} \label{eq:dSKP-R}
	\frac{(R_{(ij)} -  R_{(j)})  (R_{(jk)} - R_{(k)}) (R_{(ik)} -  R_{(i)})}{(R_{(ij)} -  R_{(i)})  (R_{(jk)} - R_{(j)}) (R_{(ik)} -  R_{(k)})}=1 .
	\end{equation}
\end{Prop}

\begin{proof}
	Equation~\eqref{eq:Z-Delta} implies that the ratio $R =Z / \tilde{Z} \in \Bbbk(x)$ satisfies 
	\begin{equation}
	\tilde{Z}_{(j)} \Delta_{(i)} (R_{(ij)} -  R_{(j)}) = \tilde{Z}_{(i)} \Delta_{(j)} (R_{(ij)} - R_{(i)}) , \qquad i<j,
	\end{equation}
	which when multiplied by its versions for pairs $j<k$ and $i<k$ 
	\begin{align}
	\tilde{Z}_{(k)} \Delta_{(j)} (R_{(jk)} -  R_{(k)}) = & \tilde{Z}_{(j)} \Delta_{(k)} (R_{(jk)} - R_{(j)}) , \qquad j<k , \\
	\tilde{Z}_{(i)} \Delta_{(k)} (R_{(ik)} -  R_{(i)}) = & \tilde{Z}_{(k)} \Delta_{(i)} (R_{(ik)} - R_{(k)}) , \qquad i<k,
	\end{align}
	gives the dSKP equations.
\end{proof}
\begin{Rem}
	Equation \eqref{eq:dSKP-R} can be solved for any of six constituent functions, in particular
	\begin{equation} \label{eq:R-HP}
	R_{(jk)} = \frac{ R_{(k)} (R_{(ij)} -  R_{(j)}) (R_{(ik)} - R_{(i)}) - R_{(j)}	(R_{(ij)} -  R_{(i)})  (R_{(ik)} -  R_{(k)}) }	
	{ (R_{(ij)} -  R_{(j)}) (R_{(ik)} -  R_{(i)}) - 
		(R_{(ij)} -  R_{(i)})  (R_{(ik)} -  R_{(k)}) } .
	\end{equation}
\end{Rem} 

In \cite{DoliwaKosiorek} it was shown that the above equation can be interpreted as a relation between six points of the projective line (here over the field $\Bbbk(x)$ of rational functions) and is equivalent to the following construction of $R_{(jk)}$ once the points $R_{(i)}$, $R_{(j)}$, $R_{(k)}$, $R_{(ij)}$ and $R_{(ij)}$ are given (see Figure \ref{fig:R-geometry}):

\begin{itemize}
	\item select any point $A$ outside the base line,
	\item on the line $\langle A, R_{(j)} \rangle$ select any point $B$ different from $A$ and $R_{(j)}$,
	\item define point $C$ as the intersection of lines $\langle A, R_{(i)} \rangle$ and $\langle B, R_{(k)} \rangle$,
	\item define point $D$ as the intersection of lines $\langle C, R_{(ik)} \rangle$ and $\langle A, R_{(ij)} \rangle$,
	\item point $R_{(jk)}$ is the intersection of the line  $\langle B, D \rangle$ with the base line.
\end{itemize}
\begin{figure}
	\begin{center}
		\includegraphics[width=9cm]{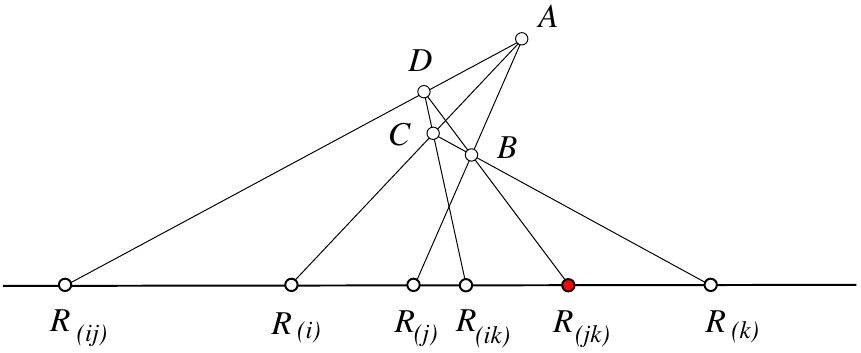}
	\end{center}
	\caption{Geometric meaning of the generalization of the Wynn recurrence (discrete Schwarzian Kadomtsev--Petviashvili equation)}
	\label{fig:R-geometry}
\end{figure}

Equation \eqref{eq:dSKP-R} can be called the generalized Wynn recurrence, see equation \eqref{eq:Wynn-R} for the standard Wynn identity. Such an interpretation will be important in discussion of the geometric initial boundary value problem. Recall from Section~\ref{sec:initial-Wynn} that on the initial plane indexed by the pair $i<j$ all other discrete parameters $n_k=-1$, $k\neq i,j$. The rational function \begin{equation}
R(n_i,n_j,x)=\frac{Z_j(-1, \dots , n_i, \dots , n_j, \dots , -1,x)}{Z_i(-1, \dots , n_i, \dots , n_j, \dots , -1,x)}
\end{equation}
 is a solution of the classical Pad\'{e} problem for the series $f(x) = -f_i(x)/f_j(x)$
\begin{equation}
R(n_i,n_j,x) = f(x) + O(x^{n_i + n_j + 1}).
\end{equation}
Notice that under such identification $n_i$ is the degree of the \emph{denominator} while $n_j$ is the degree of the \emph{numerator}.
It is well known \cite{BakerGraves-Morris} that such a rational function satisfies the Wynn identity~\cite{Wynn}
\begin{equation} \label{eq:Wynn-R}
(R_{(i)}-R)^{-1} + (R_{(-i)}-R)^{-1} = (R_{(j)}-R)^{-1} + (R_{(-j)}-R)^{-1}, 
\end{equation}
with the boundary data $R(n_i,-1,x) \equiv 0$, $R(-1, n_j,x) \equiv \infty$, for $n_i, n_j \geq 0$, and $R(0, n_j) = \lfloor f \rfloor_{n_j}$ is the polynomial approximation of $f$ up to order $n_j$. 

Interpreted as a relation between five points of the projective (base) line over the field of rational functions the recurrence gives the following construction~\cite{Doliwa-Siemaszko-W} of $R_{(j)}$ once the points $R_{(-i)}$, $R_{(-j)}$, $R$ and $R_{(i)}$ are given (see Figure \ref{fig:Wynn-geometry}):
\begin{figure}[h!]
	\begin{center}
		\includegraphics[width=9cm]{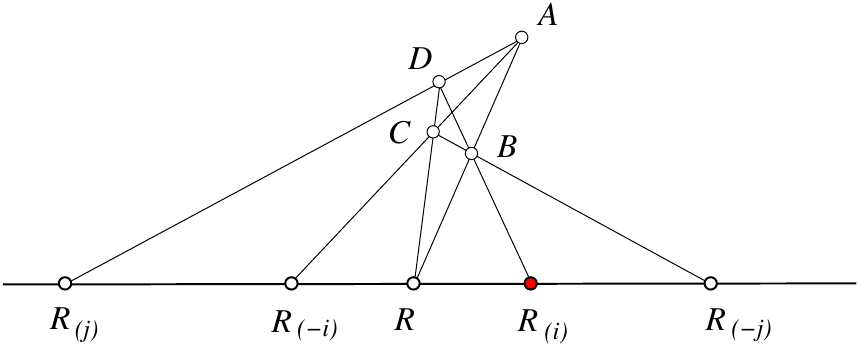}
	\end{center}
	\caption{Geometric meaning of the Wynn recurrence}
	\label{fig:Wynn-geometry}
\end{figure}
\begin{itemize}
	\item select any point $A$ outside the base line,
	\item on the line $\langle A, R \rangle$ select any point $B$ different from $A$ and $R$,
	\item define point $C$ as the intersection of lines $\langle A, R_{(-i)} \rangle$ and $\langle B, R_{(-j)} \rangle$,
	\item define point $D$ as the intersection of lines $\langle C, R \rangle$ and $\langle A, R_{(j)} \rangle$,
	\item point $R_{(i)}$ is the intersection of the line  $\langle B, D \rangle$ with the base line.
\end{itemize}
\begin{Rem}
Formally, the standard Wynn recurrence \eqref{eq:Wynn-R} can be obtained~\cite{Doliwa-Siemaszko-W} from equation \eqref{eq:dSKP-R} in three discrete variables $n_i$, $n_j$ and $n_k$ by symmetry reduction $R_{(j)} = R_{(ik)}$ (and reindexing the variables); compare Figures~\ref{fig:R-geometry} and \ref{fig:Wynn-geometry}. The same formal procedure holds~\cite{Bialecki} also for reduction from the Hirota equation~\eqref{eq:H-Delta} to the discrete-time Toda chain equation~\eqref{eq:F-DD}.
\end{Rem}

\subsection{The incidence geometric construction from the initial data}
Notice that when $n_j=-1$ then formula \eqref{eq:Z-det} implies that the corresponding homogeneous coordinate $Z_j$ equals to zero. Therefore for fixed index $i$, and $n_j = -1$ for $j\neq i$, then for all $n_i\geq 0$ we have that $P(-1, \dots , n_i, \dots ,-1)$ is the same point, which we denote by $P_i$. Similarly, for $i\neq j$ all points $P(-1,\dots , n_i, \dots , n_j , \dots , -1)$, with $n_i,n_j\geq 0$ and $n_k=-1$ for $k\neq i,j$ belong to the line $L_{ij} = \langle P_i , P_j \rangle$. 
We say that the points $P(n)$, with $|n|+m-1=N$ form $N$th level of the construction. Notice that knowledge of level $N$ points implies knowledge of lines $L(n)$ with $|n|+m=N$. 

As the initial data let us take points 
\begin{equation}
P(n_1,0,-1, \dots, -1) \in L_{12}, \quad \dots, \quad P(n_1,-1, \dots , -1 ,0) \in L_{1m}, \qquad n_1 \geq 0.
\end{equation} 
The points $P(n_1, -1, \dots , n_j, \dots , -1) \in L_{1j}$ also belong the same lines, and can be obtained by the Wynn recurrence using equation~\eqref{eq:Wynn-R} or the geometric construction visualized on Figure~\ref{fig:Wynn-geometry}. The points with three discrete variables non-negative can be obtained subsequently by the Veblen configuration construction using equation~\eqref{eq:dSKP-R} and Figure~\ref{fig:R-Veblen-PL}.

Let us present details of the construction in the basic case $m=3$. We start with three points $P_1=[1:0:0]$, $P_2=[0:1:0]$, $P_3=[0:0:1]$ and the corresponding three lines $L_{12} = \langle P_1, P_2 \rangle$, $L_{13} = \langle P_1, P_3 \rangle$ and $L_{23} = \langle P_2, P_3 \rangle$. We visualize the mechanism on example of the first three levels, see the corresponding diagrams in Figure~\ref{fig:Initial-levels}; each level is divided into rows depending on the value of the first variable. Then we present the induction step into the level $N+1$.

\medskip
\noindent Level $1$: \\
Read: $P(0,0,-1) \in L_{12}$, $P(0,-1,0) \in L_{13}$;\\*
\indent $L(0,0,0) = \langle P(0,0,-1), P(0,-1,0) \rangle$, \\*
Row $1$: \\*
\indent $P(-1,0,0) = L(0,0,0)\cap L_{23}$.\\*
----- \\
Level $2$: \\*
Read: $P(1,0,-1) \in L_{12}$, $P(1,-1,0) \in L_{13}$;\\*
\indent $L(1,0,0) = \langle P(1,0,-1), P(1,-1,0) \rangle$, \\
Row $1$:\\* 
\indent find $P(0,1,-1)$ by the Wynn recurrence construction (in $n_1$ and $n_2$ variables), \\
\indent $P(0,0,0) = L(1,0,0)\cap L(0,0,0)$ \\
\indent find $P(0,-1,1)$  by the Wynn recurrence construction (in $n_1$ and $n_3$  variables), \\
\indent $L(0,1,0) = \langle P(0,1,-1), P(0,0,0) \rangle$, \\
\indent $L(0,0,1) = \langle P(0,0,0), P(0,-1,1) \rangle$; \\
\noindent Row $2$: \\*
\indent $P(-1,1,0) = L(0,1,0)\cap L_{23}$, \\
\indent $P(-1,0,1) = L(0,0,1)\cap L_{23}$. \\*
-----















\smallskip
\noindent $\vdots$
\smallskip \\*
-----

\begin{figure}
	\begin{center}
		\includegraphics[width=6cm]{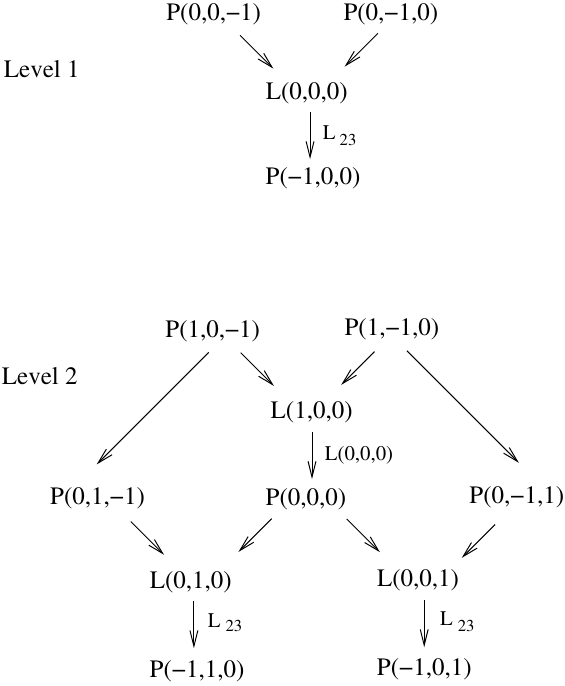} \hskip5mm 
		\includegraphics[width=8cm]{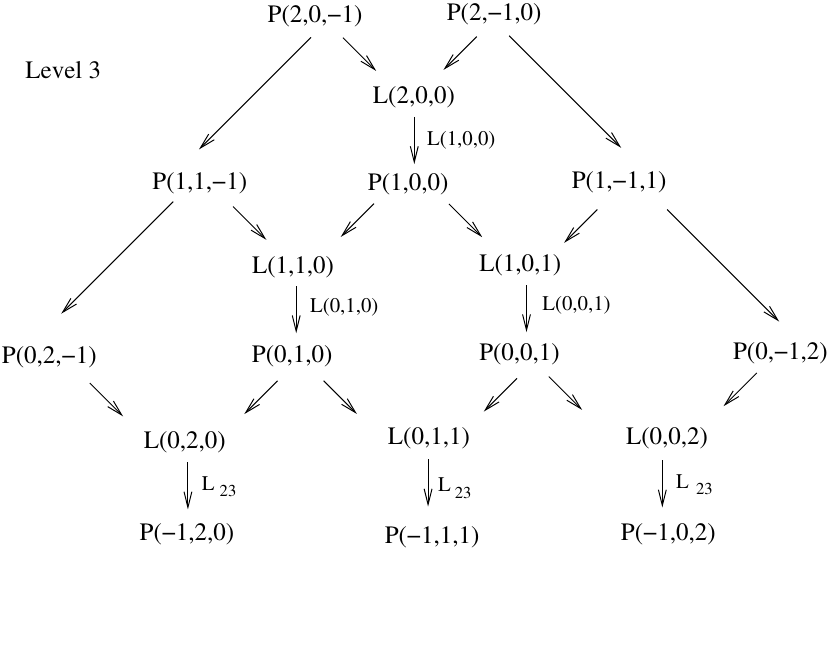}
	\end{center}
	\caption{First three levels of the geometric construction}
	\label{fig:Initial-levels}
\end{figure}

\noindent Level $N+1$:\\*
Read: $P(N,0,-1) \in L_{12}$, $P(N,-1,0) \in L_{13}$;\\*
\indent $L(N,0,0) = \langle P(N,0,-1), P(N,-1,0) \rangle$, \\*
for $k=1,\dots , N$: \\*
\indent Row $k$:\\*
\indent find $P(N-k,k,-1)$ by the Wynn recurrence construction (in $n_1$ and $n_2$ variables); \\
\indent for $\ell = 1,\dots , k$: \\*
\indent\indent $P(N-k ,k-\ell,\ell - 1 ) = L(N-k+1 ,k-\ell,\ell - 1)\cap L(N-k ,k-\ell, \ell - 1)$,\\
\indent\indent $L(N-k ,k-\ell + 1,\ell -1 ) =  
\langle P(N-k ,k-\ell +1 ,\ell - 2 ), P(N-k ,k-\ell,\ell - 1) \rangle$, \\
\indent find $P(N-k,-1,k)$ by the Wynn recurrence construction (in $n_1$ and $n_3$  variables); \\*
\indent $L(N-k ,0, k ) =  
\langle P(N-k ,0 ,k-1 ), P(N-k ,-1 ,k) \rangle$; \\
Row $N+1$:\\*
for $\ell=0,\dots ,N$: \\*
\indent $P(-1,N-\ell, \ell) = L(0,N - \ell, \ell)\cap L_{23}$.\\*
-----

\begin{Rem} 
In order to find the point $P(n_1, n_2+1,-1)$ using the geometric construction behind the Wynn recurrence (see Figure~\ref{fig:Wynn-geometry}), 
it is convenient to select for the construction as point $A=P(n_1,n_2-1,0)$. At the next step select point $B=P(n_1-1,n_2,0)$ on the line $L(n_1,n_2,0) = \langle P(n_1,n_2,-1), A \rangle$. Such a choice gives point $C=P(n_1 - 1, n_2 -1, 0)$. The corresponding point $D=D(n_1, n_2 + 1, -1)$ cannot be (generically) identified with any of the points $P(n)$. Analogous choice can be made in construction of the point $P(n_1, -1, n_3+1)$.
\end{Rem}

\begin{Rem}
The points $P(-1,n_2,n_3) \in L_{23}$ by consistency of the construction are subject to the geometric Wynn constraint (in $n_2$ and $n_3$ variables). Notice that when the polynomials $Z_1(n_1,n_2,n_3,x)$ serve as the denominators in the construction of non-homogeneous coordinats (see Section~\ref{sec:gen-Wynn}) then $L_{23}$ plays the role of the line at infinity. Therefore to see the Wynn recurrence equations~\eqref{eq:Wynn-R} on the $L_{23}$ line the initial non-homogeneus coordinate system has to be replaced by another one.
\end{Rem}

\section{Integrability of the Hirota equations with the Paszkowski constraint} \label{sec:integrability}
The determinantal structure of solutions of the Frobenius identity \eqref{eq:F-DD} corresponds, within the context of soliton theory, to \emph{semi-infinite} chain case of the 
discrete-time Toda chain equations. The equation itself is valid however for broader class of boundary conditions and its application is not restricted to the Pad\'{e} approximation only. Let us discuss therefore the equations of Section~\ref{sec:mult-eq-HP} within such more general, i.e. not restricted to the determinantal solutions and the Hermite--Pad\'{e} problem context. We go back therefore to the notation used in Section~\ref{sec:Hirota}.

According to the general approach of the soliton theory \cite{AblowitzSegur,NMPZ} nonlinear integrable equations are obtained as compatibility conditions of the corresponding system of linear equations. Moreover, admissible reductions of a given nonlinear integrable equation should arise from additional constraints imposed on the level of the linear system. Motivated by equations satisfied by the Hermite--Pad\'{e} approximants 
we expect that the $\bphi^*-\tau$ form of the Paszkowski constraint \eqref{eq:Paszkowski-Frobenius} combined with the linear system \eqref{eq:ad-lin-bil} satisfied by $\bphi^*$ should result in additional, apart from the Hirota system \eqref{eq:H-M}, nonlinear difference equation of the form \eqref{eq:Paszkowski-D} satisfied by the $\tau$--function. Moreover, by Proposition~\ref{prop:ZZ} we may foresee in such a case an additional equation satisfied by components $\phi^*_\ell$ of the wave function.

Let us start  from the basic (adjoint) linear system~\eqref{eq:ad-lin-bil} together with its compatibility condition \eqref{eq:H-M}. The following Lemma, can be verified by direct calculation. 
\begin{Lem} \label{lem:CDE}
	Define functions 
	\begin{align}
	C \, = & \tau^2 - \tau_{(1)} \tau_{(-1)} - \dots - \tau_{(m)} \tau_{(-m)}, \\
	D_\ell = & x {\phi^*_\ell} \tau - \phi^*_{\ell(1)} \tau_{(-1)} - \dots - \phi^*_{\ell(m)} \tau_{(-m)}, \qquad \ell = 1, \dots , m, \\
	E_\ell = & \phi^{*2}_\ell - \phi^*_{\ell(1)} \phi^*_{\ell(-1)} - \dots - \phi^*_{\ell(m)} \phi^*_{\ell(-m)}, 
	\end{align}
	then under assumption of the linear problem \eqref{eq:ad-lin-bil} only
	\begin{equation} \label{eq:E-D-C}
	\tau^2 E_{\ell(i)} - \tau \phi^*_\ell D_{\ell(i)} \stackrel{\eqref{eq:ad-lin-bil}}{=} \phi^{*2}_{\ell(i)} C - \tau_{(i)} \phi^*_{\ell(i)} D_\ell. 
	\end{equation}
\end{Lem}
\begin{Rem}
	In checking the Lemma it is useful to derive first the following simpler identity
	\begin{equation} \label{eq:additional-lin}
	\phi^*_{\ell(i)} C - \tau_{(i)} D_\ell \stackrel{\eqref{eq:ad-lin-bil}}{=}  \tau
	\Big( - x\phi^*_\ell \tau_{(i)} + \phi^*_{\ell(i)} \tau + \sum_{j=1}^m \mathrm{sgn}(j-i) \phi^*_{\ell(ij)} \tau_{(-j)} \Big).
	\end{equation}
\end{Rem}

\begin{Rem}
	In the context of integrable systems the parameter $x$ is called the \emph{spectral parameter}. Its precise meaning depends on the class of solutions we are interested in. 
\end{Rem}
\begin{Th} \label{th:na-mT}
The compatibility condition of the linear equations \eqref{eq:ad-lin-bil} in $m\geq 3$ discrete variables   
\begin{equation*}
\bphi^*_{(j)}\tau_{(i)}  -  \bphi^*_{(i)}\tau_{(j)} =  \bphi^*_{(ij)} \tau, 
\end{equation*}
supplemented by the (eigenvalue) equation
\begin{equation} \label{eq:lin-constr}
\bphi^*_{(1)} \tau_{(-1)} + \dots + \bphi^*_{(m)} \tau_{(-m)} = x {\bphi^*} \tau ,
\end{equation}
is the system composed of the Hirota equations \eqref{eq:H-M}
\begin{equation*} 
\tau_{(i)}\tau_{(jk)} - \tau_{(j)}\tau_{(ik)} + \tau_{(k)}\tau_{(ij)} = 0,
\qquad 1\leq i< j <k \leq m
\end{equation*}
and of the non-autonomous generalization of equation \eqref{eq:Paszkowski-D}
\begin{equation} \label{eq:H-M-T}
 \tau_{(1)} \tau_{(-1)} + \dots + \tau_{(m)} \tau_{(-m)} = F(|n|) \tau^2, 
\end{equation}
where $F$ is a function of the single variable $|n|=n_1 + \dots + n_m$.
\end{Th}
\begin{proof}
In deriving the compatibility condition for equations \eqref{eq:ad-lin-bil} and \eqref{eq:lin-constr} one can start from initial data $\bphi^*_{(-i)}$, $i=1,\dots ,m$, and using the equations arrive to $\bphi^*_{(i)}$, $i=1,\dots ,m$. Then one can check if $\bphi^*$ calculated from equation~\eqref{eq:lin-constr} coincides with that calculated using \eqref{eq:ad-lin-bil}. Notice that the Hirota equations \eqref{eq:H-M} follow from the compatibility of the linear system \eqref{eq:ad-lin-bil} only.

Assuming $D_\ell = 0$ in equations \eqref{eq:additional-lin} we get
\begin{equation} \label{eq:lin-S}
x\tau \tau_{(i)} \bphi^* = \bphi^*_{(i)} S + \tau \sum_{k=1}^m \mathrm{sgn} (k-i) \bphi^*_{(ik)} \tau_{(-k)}, 
\end{equation}
where
\begin{equation}
 S = \tau_{(1)} \tau_{(-1)} + \dots + \tau_{(m)} \tau_{(-m)}.
\end{equation}
Shifting equation \eqref{eq:lin-S}$_i$ backward in the distinguished direction $i$, and using the linear problem \eqref{eq:ad-lin-bil}, we are in principle in a position of finding all $\bphi^*_{(i)}$, $i=1,\dots ,m$ in terms of the initial data and then checking the compatibility.

Let us consider index $j>i$, shift backward equation \eqref{eq:lin-S}$_j$ in $j$th direction and multiply it by $\tau_{(-i)}/\tau_{(-j)}$. Then subtract the result from equation \eqref{eq:lin-S}$_i$  shifted back in $i$th direction and multiplied by $\tau_{(-j)}/\tau_{(-i)}$. This gives equations
\begin{gather*} 
x\tau (\tau_{(-j)}\bphi^*_{(-i)} - \tau_{(-i)}\bphi^*_{(-j)}) = \tau_{(-i)} \tau_{(-j)} \bphi^* \left( \left( \frac{S}{\tau^2} \right)_{(-i)} - \left( \frac{S}{\tau^2} \right)_{(-j)} \right) + \\
- \sum_{k<i,j}\bphi^*_{(k)} \left( \tau_{(-j)} \tau_{(-i-k)} - \tau_{(-i)} \tau_{(-j-k)} \right) + \bphi^*_{(i)} \tau_{(-i)} \tau_{(-i-j)} + 
\sum_{i<k<j}\bphi^*_{(k)} \left( \tau_{(-j)} \tau_{(-i-k)} + \tau_{(-i)} \tau_{(-j-k)} \right) + \\
+ \bphi^*_{(j)} \tau_{(-j)} \tau_{(-i-j)} + 
\sum_{i,j<k}\bphi^*_{(k)} \left( \tau_{(-j)} \tau_{(-i-k)} - \tau_{(-i)} \tau_{(-j-k)} \right) ,
\end{gather*}
which because of the linear problem \eqref{eq:ad-lin-bil} and the Hirota system \eqref{eq:H-M} reduce to
\begin{equation*}
x\tau \tau_{(-i-j)} \bphi^* = 
\tau_{(-i)} \tau_{(-j)} \bphi^* \left( \left( \frac{S}{\tau^2} \right)_{(-i)} - \left( \frac{S}{\tau^2} \right)_{(-j)} \right) + \tau_{(-i-j)} \sum_{k=1}^m \tau_{(-k)} \bphi^*_{(k)}.
\end{equation*}
Finally, equation~\eqref{eq:lin-S} implies that the compatibility condition takes the following simple form
\begin{equation*}
\left( \frac{S}{\tau^2} \right)_{(-i)} = \left( \frac{S}{\tau^2} \right)_{(-j)} , \qquad i,j = 1, \dots , m,
\end{equation*}
i.e. $S/\tau^2$ is a function depending on the sum $|n|= n_1 + \dots + n_m$ of all variables.
\end{proof}
\begin{Cor}
	Generically, when the function $F$, given in equation \eqref{eq:H-M-T}, does not vanish, then the gauge function $G$ defined by
\begin{equation}
G(k) = \prod_{i=1}^{k_0} F(k-i)^i,
\end{equation}
can be used to define the rescaled $\tau$-function $\tilde{\tau}$
\begin{equation}
\tilde{\tau}(n)  = \frac{\tau(n)}{G(|n|)},
\end{equation}
which satisfies the Hirota system~\eqref{eq:H-M} with the original form of equation \eqref{eq:Paszkowski-D}
\begin{equation*}
\tilde{\tau}_{(1)} \tilde{\tau}_{(-1)} + \dots + \tilde{\tau}_{(m)} \tilde{\tau}_{(-m)} =  \tilde{\tau}^2.
\end{equation*}
\end{Cor}
\begin{Cor}
	By Corollary~\ref{cor:quad-phi} and Lemma~\ref{lem:CDE} the components $\phi_\ell^*$ of the adjoint wave function satisfy equations
\begin{align} 
\phi^*_{\ell(i)}\phi^*_{\ell(jk)} - \phi^*_{\ell(j)}\phi^*_{\ell(ik)} + \phi^*_{\ell(k)}\phi^*_{\ell(ij)} = 0,
\qquad & 1\leq i< j <k \leq m \\
\phi^*_{\ell(1)} \phi^*_{\ell(-1)} + \dots + \phi^*_{\ell(m)} \phi^*_{\ell(-m)} = F(|n|-1) \phi^{*2}_\ell .& 
\end{align}
\end{Cor}

\begin{Rem}
In \cite{AptekarevDerevyaginMikiVanAssche} the following system (equations (3.32) and (3.33) of that paper) consisting of the Hirota(--Miwa) equations
\begin{equation} \label{eq:HM-ADMV}
\tau^t_{\vec{n}+ \vec{e}_i + \vec{e}_j} \tau^{t+1}_{\vec{n}} -
\tau^t_{\vec{n}+ \vec{e}_i } \tau^{t+1}_{\vec{n}+ \vec{e}_j} + 
\tau^t_{\vec{n}+ \vec{e}_j } \tau^{t+1}_{\vec{n}+ \vec{e}_i }
= 0, \qquad i\neq j,
\end{equation}
supplemented by the constraint
\begin{equation}
\label{eq:P-ADMV}
\tau^{t+1}_{\vec{n}}\tau^{t-1}_{\vec{n}} = 
(\tau^{t}_{\vec{n}})^2 + \sum_{k=1}^r 
\tau^{t-1}_{\vec{n}+ \vec{e}_k } 
\tau^{t+1}_{\vec{n} - \vec{e}_k }, 
\end{equation} 
is obtained within the context of multiple orthogonal polynomials; here $\vec{n} = (n_1, \dots , n_r)$ and $\vec{e}_i$, $i=1, \dots , r$, is the $i$th vector of the standard basis. The above equations constitute bilinear form for
the so called \emph{discrete multiple} Toda lattice, obtained by generalization of the orthogonal polynomial approach to  the discrete Toda lattice. Under identification (abusing slightly the notation)
\begin{equation}
\tau^{t}_{\vec{n}} = i^{n_1^2 + \dots + n_r^2} \tau(n_1,\dots, n_r,n_{r+1}),
\end{equation}
where 
\begin{equation} 
n_{r+1} = - (t + n_1 + \dots + n_r), \qquad m = r+1,
\end{equation}
one obtains equations the Hirota system supplemented by the Paszkowski constraint studied in the present paper.

The discrete linear system (the Lax set) considered in \cite{AptekarevDerevyaginMikiVanAssche} is obtained from the Christoffel and Geronimus transformations of multiple orthogonal polynomials, which give the discrete-time evolution. Both equations~\eqref{eq:HM-ADMV} and \eqref{eq:P-ADMV} involve variation of the time variable, which seems to be essential in the integrability construction presented there. 
\end{Rem}

\bibliographystyle{amsplain}

\section{Conclusion}
Although the Hermite--Pad\'{e} approximants had been introduced by Hermite in his proof of transcendency of the Euler constant $e$ they have found applications not only in number theory, but also in the theory of multiple orthogonal polynomials, random matrices or numerical analysis.
Our result, which incorporates the theory of Hermite--Pad\'{e} approximants into the theory of integrable systems explains the appearance of various ingredients of the integrable systems theory in all the branches of mathematical physics and applied mathematics where the Hermite--Pad\'{e} approximation problem is relevant.

It is interesting therefore to notice that Hirota's discrete KP equation, which is considered as the most important discrete integrable equation, appeared almost at the same time in the theory of integrable systems (for $m=3$) and in the theory of Hermite--Pad\'{e} approximants (for arbitrary $m\geq 3$).
It turns out that in the application to the Hermite--Pad\'{e} approximants the Hirota system should be supplemented by the Paszkowski constraint, which provides an integrable reduction of the system. We have not found any mention of his  papers in research articles devoted to relation of the Hermite--Pad\'{e} approximants, orthogonal polynomials or random matrices to integrability. 

One can think about application of techniques from the integrability theory (finite gap integration, non-commutative generalizations, etc.) to study corresponding versions of the Hermite--Pad\'{e} approximants. On the other hand in the numerical analysis there are known other nonlinear methods~\cite{CuytWuytack} which generalize the approximants and deserve a closer look from the point of view of integrable systems.

\section*{Acknowledgements}
We thank the Reviewers for constructive comments which helped us to improve the presentation.

\providecommand{\bysame}{\leavevmode\hbox to3em{\hrulefill}\thinspace}

\end{document}